\documentclass[11pt]{article}
\usepackage{fullpage}
\usepackage{epsfig}
\usepackage{graphics}
\usepackage{latexsym}
\usepackage{amsmath}
\usepackage{amsfonts}
\usepackage{amssymb}
\usepackage{mathrsfs}
\usepackage{amsthm}
\usepackage{xspace}
\usepackage{epstopdf}
\usepackage{float}
\usepackage{hyperref}
\usepackage{mathtools}

\numberwithin{equation}{section}
\usepackage[ruled,vlined]{algorithm2e}
\SetArgSty{textrm}

\usepackage[english]{babel}
\usepackage[nottoc]{tocbibind}

\usepackage{caption}
\usepackage{subcaption}

\usepackage[usenames,dvipsnames]{color}
\usepackage{pifont}
\usepackage{booktabs}
\usepackage{multirow}

\newcommand{\ma}[1]{\textcolor{magenta}{\textrm{#1}}}
\newcommand{\redcomment}[1]{\textcolor{red}{\textrm{#1}}}

\usepackage{amsthm}     
\theoremstyle{plain}     
\newtheorem{theorem}{Theorem}

\newtheorem{lemma}{Lemma}
\newtheorem{proposition}{Proposition}
\theoremstyle{definition} 

\newtheorem{mechanism}{Mechanism}
\theoremstyle{remark} 

\makeatletter
\@addtoreset{equation}{section}
\def\section{\@startsection {section}{1}{\z@}{-3.5ex plus -1ex minus
 -.2ex}{2.3ex plus .2ex}{\large\bf}}
\makeatother

\def\bfm#1{\mbox{\boldmath$#1$}}

\def\0{\bfm 0}

\DeclareMathAlphabet{\mathpzc}{OT1}{pzc}{m}{it}

\newcounter{my}

\newcounter{my2}

\newcounter{my3}

\newcounter{my4}

\newcounter{my5}

\newcounter{my6}

\allowdisplaybreaks 

\begin{document}

\title{Obnoxious Facility Location Problems: Strategyproof Mechanisms Optimizing $L_p$-Aggregated Utilities and Costs}


\date{}
\maketitle

\vspace{-3em}
\begin{center}

\author{Hau Chan$^{1}$\quad Jianan Lin$^{2}$\quad Chenhao Wang $^{3,4}$\\
${}$\\
$1$ University of Nebraska-Lincoln\\
$2$ Rensselaer Polytechnic Institute\\
$3$ Beijing Normal University-Zhuhai\\
$4$ IRADS, Beijing Normal-Hong Kong Baptist University\\
\medskip
}

\end{center}

\begin{abstract}
We study the problem of locating a single obnoxious facility on the normalized line segment $[0,1]$ with strategic agents from a mechanism design perspective. 
Each agent has a preference for the undesirable location of the facility and would prefer the facility to be far away from their location. 
We consider the utility of the agent, defined as the distance between the agent's location and the facility location, and the cost of each agent, equal to one minus the utility. 
Given this standard setting of obnoxious facility location problems, our goal is to design (group) strategyproof mechanisms to elicit agent locations truthfully and determine facility location approximately optimizing the $L_p$-aggregated utility and cost objectives, which generalizes the $L_p$-norm ($p\ge 1$) of the agents' utilities and agents' costs to any $p \in [-\infty, \infty]$, respectively. 
We establish upper and lower bounds on the approximation ratios of deterministic and randomized (group) strategyproof mechanisms for maximizing the $L_p$-aggregated utilities or minimizing the $L_p$-aggregated costs across the range of \(p\)-values. 
While there are gaps between upper and lower bounds for randomized mechanisms, our bounds for deterministic mechanisms are tight. 
\end{abstract}

\section{Introduction}\label{sec:intro}

Over the past several decades, facility location problems have been extensively studied theoretically and practically in various fields, including operations research \cite{feigenbaum2017approximately,owen1998strategic}, computer science \cite{procaccia2013approximate,sonoda2016false}, and economics \cite{aziz2020capacity,d1979hotelling}. 
In a typical facility location problem, a social planner seeks to locate one or more facilities to \emph{best} serve a set of agents within a region, where \emph{best} is often understood as optimizing a specific utility or cost objective based on the distances between the agents' preferred facility locations and the locations of the facilities. 
Classic practical applications of facility location problems include selecting a location for a public school, library, or hospital to serve the surrounding communities, while taking into consideration the facility location preferences of the agents in those communities. 

\paragraph{Obnoxious Facility Location Problems}
In the most-studied facility location problems, such as locating a library, park, or transit station, agents often view these facilities as desirable and prefer them to be located closer to their preferred locations. 
In contrast, a social planner might sometimes locate obnoxious facilities (e.g., nuclear power plants, landfill sites, chemical factories, and detention centers) that are essential for society's functioning but are viewed as undesirable or unpleasant by the agents. 
For instance, landfill sites may expose residents to unpleasant odors and hazardous gases (e.g., lead, ammonia, and hydrogen sulfide) \cite{atsdr1997agency, willumsen1990landfill}. 
Nuclear radiation from power plants has been associated with higher risks of cancer and leukemia \cite{kaatsch2008leukaemia, baker2007meta}. 
Naturally, the agents prefer these obnoxious facilities to be located far away from their undesirable facility locations. 
Therefore, another active area of research \cite{drezner2004facility, eiselt2011foundations} 
examines obnoxious facility location problems that deal with locating obnoxious facilities optimizing the objective based on their undesirable facility location preferences. 

\paragraph{Mechanism Design for Obnoxious Facility Location Problems} 
Because agents' undesirable location preferences are often unknown, there have been significant efforts to design mechanisms to elicit agent preferences and use these preferences to determine optimal facility locations based on the objective. 
However, it is well known that each agent may have an incentive to misreport their true preference to manipulate the mechanisms' locations to benefit themselves \cite{moulin1980strategy,procaccia2013approximate,chan2021mechanismsurvey}. 
Therefore, a major multidisciplinary research area in mechanism design for obnoxious facility location problems  \cite{Ibara2012characterize,cheng2013strategy,feigenbaum2015strategyproof,ye2015strategy,alex2024,li2024strategyproof,DBLP:journals/corr/abs-2212-09521,chan2021mechanismsurvey} examines the design of strategyproof mechanisms that elicit agents' undesirable location preferences truthfully on a line segment and determine facility locations that (approximately) optimize a specific utility or cost objective. 
The most commonly studied utility objectives include the total utility \cite{cheng2013strategy,feigenbaum2015strategyproof,ye2015strategy}, minimum utility \cite{feigenbaum2015strategyproof}, and sum of squared utilities \cite{ye2015strategy}, which are defined based on the agents' (individual) utilities, where each agent's utility is defined to be the distance between the facility location and the agent's undesirable location. 
For the cost objectives where each agent incurs a cost defined as the segment length minus their utility, existing studies examine the total cost and maximum cost \cite{li2024strategyproof}.

Because existing obnoxious facility location studies focus on designing specialized strategyproof mechanisms optimizing these utility or cost objectives separately, we investigate the extent to which we can design unified strategyproof mechanisms optimizing a family of \emph{\(L_p\)-aggregated} utility or cost objectives for $p \in [-\infty, \infty]$. 
The \(L_p\)-aggregated objectives are generalization of  \(L_p\)-norm ($p\ge 1$) on the agents' utilities or costs, and contain the above mentioned objectives as special cases (e.g., total utility when $p = 1$ and minimum utility when $p =-\infty$). 
That is, the \(L_p\)-aggregated objectives encapsulate a spectrum of social objectives, where varying $p$ mediates the trade-off between utilitarian efficiency and egalitarian fairness \cite{feigenbaum2017approximately}. Specifically, when $p$ increases, higher costs have a disproportionately larger influence on the objective, and the social planner gives more weights to agents who experience larger costs.  
In practice, the social planner can choose different values of $p$ to implement certain fairness and trade-offs.

While the standard \(L_p\)-norm cost objectives have been examined in the mechanism design for classic (desirable) facility location problems \cite{feigenbaum2017approximately,GoelH23,lin2020nearly,chan2025prediction,gravin2025approximation,DBLP:conf/sagt/Meir19,DBLP:conf/nips/Barak0T24}, we are the first to explore the $L_p$-norm objectives, and more generally, the $L_p$-aggregated utility and cost objectives with any value of $p$, in mechanism design for obnoxious facility location.

\subsection{Our Contributions}

We study the design of strategyproof (SP) and group strategyproof (GSP) mechanisms for the standard setting  \cite{Ibara2012characterize,cheng2013strategy,feigenbaum2015strategyproof,ye2015strategy,alex2024,li2024strategyproof,DBLP:journals/corr/abs-2212-09521,chan2021mechanismsurvey} of obnoxious facility location problems with $n$ agents under \(L_p\)-aggregated social utility (su) and social cost (sc) objectives. 
In this setting, a social planner aims to locate a single obnoxious facility on a normalized line segment $[0,1]$, and each agent $i$ has an undesirable facility location at $x_i \in [0,1]$. 
We provide approximation guarantees or ratios for the designed SP/GSP mechanisms and any SP/GSP mechanisms in terms of upper bounds and lower bounds, respectively, across the range of \(p\)-values, including boundary cases such as \(p = \pm\infty\) and \(p \to 0^+\).
Table~\ref{table:result} summarizes our upper and lower bounds using intervals or single numbers. 
All of our bounds for deterministic SP/GSP mechanisms are tight.

\paragraph{Social Utility Objectives.} 
Under the su objective, the utility of each agent at $x_i$ is their distance to the facility location $y$. 
We consider maximizing the $L_p$-aggregated social utility objective $\sum_{i=1}^n(|x_i-y|^p)^{1/p}$. 
For deterministic SP/GSP mechanisms, we provide a tight bound of \(2\) for maximum utility (the case \(p = +\infty\)) and a tight bound of \((2^p+1)^{1/p}\) for all \(0<p<\infty\). 
In the limit \(p \to 0^+\) (which corresponds to Nash welfare up to a monotone transformation via the geometric mean), we show that the approximation ratio is unbounded. 
For randomized SP/GSP mechanisms, we show the following: 
for maximum utility, an upper bound of \(4/3\) and a lower bound of \(6/5\); 
for \(1 \le p < \infty\), an upper bound of \( \left(\frac{2^p+1}{2^{p-1}+1}\right)^{1/p}\) and a lower bound of \( \left(\frac{4(3^p+1)}{3(3^p+1)+2}\right)^{1/p}\); 
for \(0<p<1\), an upper bound of \(2^{1/p}\) with the same lower bound \( \left(\frac{4(3^p+1)}{3(3^p+1)+2}\right)^{1/p}\); 
for \(p \to 0^+\), an upper bound of \(\sqrt{2}+1\) and a lower bound of \(\sqrt{6/5}\); and 
for minimum utility (the egalitarian objective, \(p = -\infty\)), an upper bound of \(\mathcal{O}(\sqrt{n})\), together with a lower bound of at least \(3/2\) in the limit \(n \to \infty\) following~\cite{feigenbaum2015strategyproof} and a specialized lower bound of approximately \(1.026\) for \(n=2\). 

\paragraph{Social Cost Objectives.}
Under the sc objective, the cost of each agent at $x_i$ is one minus their distance to the facility location $y$ \cite{li2024strategyproof}. 
We consider minimizing the $L_p$-aggregated social cost  objective $\sum_{i=1}^n((1-|x_i-y|)^p)^{1/p}$. 
For \(p<1\), the \(L_p\)-aggregated sc is a non-convex, non-norm aggregator that over-rewards concentrating harm on a few agents (risk-seeking), violating fairness and robustness. Practically, this undermines geometric intuition and lacks a clear welfare interpretation. 
Hence, we focus on \(p \ge 1\).
For deterministic SP/GSP mechanisms, we provide a tight bound of \(2\) for maximum cost (the egalitarian objective, \(p= +\infty\)) and a tight bound of \((2^p+1)^{1/p}\) for all \(1 \le p < \infty\). 
For randomized SP/GSP mechanisms, we obtain the following: for maximum cost, a lower bound of approximately \(1.008\); for \(1 \le p < \infty\), an upper bound of \(2\) when \(p=1\) and a general lower bound of \((5/4)^{1/p}\). 
Moreover, we show that no \emph{two-candidate} randomized mechanism can achieve a ratio better than \((2^{p-1}+1)^{1/p}\) for any \(p \/ge 1\), including the maximum-cost limit \(p =+\infty\).

\begin{table}[h]
    \centering 
	\caption{Main Upper and Lower Bound Results of SP/GSP Deterministic and Randomized Mechanisms.}
	\label{table:result}
	\begin{tabular}{rll}
    \toprule
		\textit{Objectives} & \textit{Deterministic} & \textit{Randomized}\\ 
        \midrule
		su ($p=+\infty$) & $2$ & $\left[\frac{6}{5}, \frac{4}{3}\right]$ \\
		su ($1\le p<\infty$) & $(2^{p}+1)^\frac{1}{p}$ & $\left[\left(\frac{4(3^p+1)}{3(3^p+1)+2}\right)^{\frac{1}{p}}, \left(\frac{2^p+1}{2^{p-1}+1}\right)^{\frac{1}{p}}\right]$ \\
		su ($0<p<1$) & $(2^{p}+1)^\frac{1}{p}$ & $\left[\left(\frac{4(3^p+1)}{3(3^p+1)+2}\right)^{\frac{1}{p}}, 2^{\frac{1}{p}}\right]$ \\
		su ($p\rightarrow 0^+$) & unbounded & $\left[\sqrt{\frac{6}{5}}, \sqrt{2}+1\right]$ \\
		su ($p=-\infty$) & unbounded \cite{feigenbaum2015strategyproof} & $[\frac{3}{2}$ \cite{feigenbaum2015strategyproof}$, \mathcal{O}(\sqrt{n})]$ \\ \hline\hline
		sc ($p=+ \infty$) & $2$ & $\ge 1.008$ \\
        sc ($1\le p<\infty$) & $(2^{p}+1)^\frac{1}{p}$ & $\ge(\frac{5}{4})^{\frac{1}{p}}$ \\
        sc ($p=1$) & $3$ & $\left[\frac{5}{4}, 2\right]$ \\
        \bottomrule
	\end{tabular}
\end{table}

\paragraph{Outline.}
After discussing the related work below, the paper is structured as follows. 
Section~\ref{sec:model} presents the preliminaries, including the basic setting, objective functions, and the SP/GSP mechanisms under consideration. Sections~\ref{sec:su} and~\ref{sec:sc} analyze the approximation bounds of SP/GSP mechanisms for $L_p$-aggregated social utilities and social costs, respectively.

\subsection{Related Work}
We review mechanism design studies most relevant to obnoxious facility location problems and the consideration of $L_p$-norm social costs in (desirable) facility location problems. We refer readers to \cite{chan2021mechanismsurvey} for a comprehensive survey on mechanism design for general facility location problems, e.g., \cite{alon2010strategyproof,lu10mechanism,DBLP:conf/sagt/Meir19,GoelH23,feldman2013strategyproof}, which are first considered by \cite{moulin1980strategy,procaccia2013approximate}. 

\paragraph{Mechanism Design for Obnoxious Facility Location Problems.} 
Most studies focus on single obnoxious facility on a normalized interval.
\cite{Ibara2012characterize} characterize deterministic SP mechanisms as two-candidate mechanisms.  
\cite{cheng2013strategy} initiate the study of approximate mechanism design for this setting, aiming to maximize total utility (i.e., $p=1$).
They give a 3-approximation deterministic mechanism and a 1.5-approximation randomized mechanism. 
\cite{feigenbaum2015strategyproof} show the tightness of the above bounds. 
For the minimum-utility objective (i.e., $p=-\infty$), they further show that no deterministic strategyproof mechanism achieves a bounded approximation ratio, while any randomized mechanism has a ratio of at least 1.5. 
Ye et al. \cite{ye2015strategy} study both objectives
of maximizing total utility and the sum of squared utilities; 
their results for randomized mechanisms imply an upper bound of $\sqrt{5/3}$ and a lower bound of 1.021 for the $L_2$-norm social utility—where our work improves the latter. The tree network is studied by Oomine and Nagamochi \cite{oomine2016characterizing}. 
Some recent studies consider fairness \cite{alex2024, li2024strategyproof} and mechanism design with  “predictions” for locating a single obnoxious facility \cite{DBLP:journals/corr/abs-2212-09521}. 

\paragraph{$L_p$-norm Social Costs.} 
While the $L_p$-norm utility or cost objective has not been considered in obnoxious facility location problems, it has been considered in the classic (single desirable) facility location problems. 
On the real line, \cite{feigenbaum2017approximately} show that the median mechanism achieves a tight $2^{1-1/p}$-approximation for minimizing $L_p$-norm social cost, and they also provide some results for randomized mechanisms 
In the 2-d Euclidean space $\mathbb R^2$, \cite{GoelH23} analyze the $L_p$-norm social cost objective and show that the coordinate-wise median (CM) mechanism has the best approximation ratio among deterministic, anonymous, strategyproof mechanisms. 
For $p=1$ with an odd number of agents \(n\),  CM attains an approximation ratio of \(\sqrt{2}\,\frac{\sqrt{n^2+1}}{n+1}\). For \(p \ge 2\), its ratio is bounded above by \(2^{\frac{3}{2}-\frac{2}{p}}\).

Another line of work studies $d$-dimensional $L_p$ spaces, where agent costs (rather than the social objective) are defined via $L_p$-norm  distances. 
\cite{lin2020nearly} characterizes the two-agent SP mechanisms in $L_p$ spaces. 
\cite{chan2025prediction} show that CM mechanism is 2-approximation for the maximum cost in any two-dimensional $L_p$ space. 
\cite{gravin2025approximation} show that, in any $L_p$ space, CM achieves at most a 3-approximation for minimizing total cost, improving upon the previously best $\sqrt d$ approximation \cite{DBLP:conf/sagt/Meir19,DBLP:conf/nips/Barak0T24}.

\section{Preliminaries}\label{sec:model}

There are $n$ agents $N=\{1,2,\ldots,n\}$ and an interval $[0, 1]$, where the distance of points $a,b\in[0,1]$ is $d(a,b)=|a-b|$. 
Each agent $i\in N$ reports a location $x_i \in [0, 1]$ which is their preferred (undesirable) location.
We want to determine the location $y\in [0, 1]$ of a single obnoxious facility. 
A \emph{deterministic} mechanism is a function $f: [0, 1]^n \rightarrow  [0, 1]$ that maps each location profile $\mathbf{x}=(x_1,\ldots,x_n)$ to a facility location. 
A \emph{randomized} mechanism is a function $f : [0, 1]^n \rightarrow \Delta([0, 1])$ that maps each location profile $\mathbf x$ to a probability distribution over $[0,1]$, where the facility is located by sampling from the distribution $f(\mathbf x)$.

In the setting of strategic agents, each agent's preferred undesirable location is private information, and mechanisms are required to be \emph{strategyproof}—for every agent $i$, reporting their \emph{true} preferred location $x_i$ a (weakly) dominant strategy. Given facility location $y$, we assume that each agent $i$ either has a utility $u(x_i, y)=d(x_i,y)=|x_i-y|$ that equals to their distance to $y$ (see \cite{cheng2013strategy,feigenbaum2015strategyproof}), or incurs a cost $c(x_i, y) = 1 - |x_i-y|$ that equals to the interval length minus the distance (see \cite{li2024strategyproof}). If the facility location is random and follows a distribution $\pi$, then the utility and cost of agent $i$ are simply the expectations $u(x_i,\pi)=\mathbb E_{y\sim \pi}u(x_i,y)$ and $c(x_i,\pi)=\mathbb E_{y\sim \pi}c(x_i,y)$. 

A mechanism $f$ is \emph{strategyproof} (SP) if for all $\mathbf x\in [0, 1]^n ,i\in N, x_{i}'\in [0, 1]$, we have $d(x_i,f(\mathbf x))\ge d(x_i,f(x_i',\mathbf x_{-i}))$, where $\mathbf x_{-i}$ is the profile of all agents but $i$.
Further, a mechanism is \emph{group strategyproof} (GSP) if no group of agents can collude to misreport in a way that makes every member better off. Formally,  $f$ is GSP if for all $\mathbf x\in [0, 1]^n,S\subseteq N, \mathbf x_{S}'\in[0, 1]^{|S|}$, there exists $i\in S$ so that $d(x_i,f(\mathbf x))\ge d(x_i,f(\mathbf x_S', \mathbf x_{-S}))$.

\subsection{$L_p$-aggregated Objectives}

The mechanisms are evaluated by the worst-case performance on some social objectives. In this paper we consider $L_p$-aggregated objectives in terms of both social utility (su) and social cost (sc).

\paragraph{Social utility objectives.}

Given a facility location $y$ when the profile is $\mathbf x=(x_1,\ldots,x_n)$, 
for any $p\in \mathbb R$, we define the \emph{$L_p$ social utility}  as 
\begin{align*}
      &\text{su}_p(y, \mathbf{x})= \left(\sum_{i\in N}u(x_i, y)^p\right)^{\frac{1}{p}}.
\end{align*}
If $y$ follows a probability distribution $\mathcal P$, then the $L_p$ social utility is the expectation $\mathbb E_{y\sim\mathcal P}su_p(y,\mathbf x)$. 

Our goal now is to find SP/GSP mechanisms that (approximately) maximize the social utility. Let $OPT_p(\mathbf x)=\max_{y\in[0,1]}\text{su}_p(y, \mathbf{x})$ be the optimal $L_p$ social utility. A mechanism is said to be \emph{$\alpha$-approximation} (or have an approximation ratio $\alpha$) for the $L_p$ social utility, if $\frac{OPT_p(\mathbf x)}{\text{su}_p(f(\mathbf x), \mathbf{x})}\le \alpha$ for all profile $\mathbf x\in[0,1]^n$. 

One closely related term in mathematics is called the \emph{power mean}: the $p$-power mean of utilities is $\left(\sum_{i\in N}\frac{1}{n}\cdot u(x_i, y)^p\right)^{1/p}$, which differs with $\text{su}_p(y, \mathbf{x})$ by a normalization constant $\frac{1}{n^{1/p}}$.
Hence, a mechanism has the same approximation ratio for maximizing the $L_p$ social utility and the $p$-power mean.

We interpret the \(L_p\) social utilities \(\mathrm{su}_p\) as follows. For \(p \ge 1\), \(\mathrm{su}_p\) is the \(L_p\)-norm of the utility vector, and for \(0<p<1\) it is the \(L_p\) quasi\mbox{-}norm. At \(p=1\), \(\mathrm{su}_p\) coincides with the utilitarian objective (the sum of utilities). For finite negative \(p\), \(\mathrm{su}_p\) is generally ill\mbox{-}defined unless all utilities are strictly positive, since negative exponents at zero are undefined. As \(p = +\infty\), it is the max\mbox{-}utility objective, whereas as \(p = -\infty\), it is the egalitarian (min\mbox{-}utility) objective. 
A notable boundary case occurs as  $p\to 0^+$: the unnormalized $\text{su}_p$ becomes unbounded, so we instead consider the $p$-power mean $\lim\limits_{p\to 0^+}\left(\sum_{i\in N}\frac{1}{n}\cdot u(x_i, y)^p\right)^{1/p}=\sqrt[n]{\prod_{i\in N} u(x_i, y)}$, which is the geometric mean of utilities. Maximizing this  geometric mean is equivalent to maximizing the Nash welfare $\prod_{i\in N} u(x_i, y)$. 

\paragraph{Social cost objectives.}

For any $p\in \mathbb R$, given facility location $y$ and profile $\mathbf x$, define the \emph{$L_p$ social cost}  as 
\begin{align*}
    \text{sc}_p(y, \mathbf{x}) &= \left(\sum_{i\in N}c(x_i, y)^p\right)^{\frac{1}{p}}.
\end{align*}
If $y$ follows a distribution $\mathcal P$, then the $L_p$ social cost is $\mathbb E_{y\sim\mathcal P}sc_p(y,\mathbf x)$. 

We want to minimize the social cost. Let the optimal  $L_p$ social cost be $OPT_p(\mathbf x)=\min_{y\in[0,1]}\text{sc}_p(y, \mathbf{x})$. A mechanism is said to be \emph{$\alpha$-approximation} for the $L_p$ social cost, if $\text{sc}_p(f(\mathbf x), \mathbf{x})\le \alpha\cdot OPT_p(\mathbf x)$ for all profile $\mathbf x\in[0,1]^n$.  
The interpretations of the $L_p$ social costs parallel those of $L_p$ social utilities, with the roles of min and max reversed: in particular, as $p=+\infty$, it is the  egalitarian (max-cost) objective. 

\subsection{GSP Mechanisms}

For a deterministic mechanism $f$, a point $y$ is  called a \emph{candidate} if there is a profile $\mathbf x$ such that $f(\mathbf x)=y$. 
Ibara and Nagamochi \cite{Ibara2012characterize} characterize all deterministic SP/GSP mechanism as  1-candidate mechanisms (that return a fixed point) and 2-candidate \emph{valid threshold} mechanisms. Roughly, a 2-candidate valid threshold mechanism chooses between the two candidates by comparing the number of agents closer to one candidate against a tie-aware cutoff—one that stays within feasible bounds and shifts by at most one when an additional agent is exactly tied; crossing the cutoff switches the outcome (see the detailed definition in \cite{Ibara2012characterize}).

\begin{lemma}\label{lemma:2-candidate}[\cite{Ibara2012characterize}]
Every 2-candidate valid threshold mechanism is GSP, and every deterministic SP mechanism is either a 2-candidate valid threshold mechanism or a 1-candidate mechanism.
\end{lemma}

For convenience, we suppose $x_1\le x_2\le \ldots \le x_n$. 
Given profile $\mathbf{x}$, let $n_1$ be the number of agents with $x_i\in [0, \frac{1}{2}]$ and $n_2$ be the number of agents with $x_i\in (\frac{1}{2}, 1]$.
We consider the following  mechanisms: the first is deterministic, and all others are randomized.

\begin{mechanism}\label{mec:mv}[{\rm Majority Vote} \cite{cheng2013strategy}]
  Given profile $\mathbf{x}$, if $n_1\le n_2$, return $y=0$; otherwise return $y=1$.
\end{mechanism}

\begin{mechanism}\label{mec:uniform}
    Return the uniform distribution over interval $[0,1]$. 
\end{mechanism}

\begin{mechanism}\label{mec:mv-ran}
    Given profile $\mathbf{x}$, return $y=0$ with probability $\frac{n_2^2}{n_1^2+n_2^2}$ and $y=1$ with probability $\frac{n_1^2}{n_1^2+n_2^2}$.
\end{mechanism}

\begin{mechanism}\label{mec:mv-ran2}
    Given profile $\mathbf{x}$, return $y=0$ with probability $P_0 = \frac{n_2^2+2^pn_1n_2}{n_1^2+n_2^2+2^{p+1}n_1n_2}$ and $y=1$ with probability $1-P_0=\frac{n_1^2+2^pn_1n_2}{n_1^2+n_2^2+2^{p+1}n_1n_2}$.
\end{mechanism}

\begin{lemma}
    Mechanism \ref{mec:mv}-\ref{mec:mv-ran2} are all GSP.
\end{lemma}

\begin{proof}
    Mechanism \ref{mec:mv} is GSP as it is a 2-candidate valid threshold mechanism \cite{cheng2013strategy,Ibara2012characterize}. 
    Mechanism \ref{mec:uniform} is GSP  because the output is a fixed distribution.
 
    For Mechanism \ref{mec:mv-ran}, suppose group $S$ misreports locations and assume w.l.o.g. the probability of $y=0$ decreases. In this way, only those agents with $x_i< \frac{1}{2}$ are better off and can join the group. However, without the attendance of agents located in $(\frac12,1]$, $n_1$ cannot increase, and thus the probability of $y=0$, i.e. $\frac{(n-n_1)^2}{n_1^2+(n-n_1)^2}$ would not decrease (which is a decreasing function of $n_1$). This gives a contradiction to the assumption. The proof for Mechanism \ref{mec:mv-ran2} follows from the same arguments, as $P_0$ is decreasing with $n_1$. 
\end{proof}

\section{$L_p$ Social Utilities}\label{sec:su}

The utility objectives have been widely studied in the literature,  including the total utility \cite{cheng2013strategy,feigenbaum2015strategyproof,ye2015strategy}, minimum utility \cite{feigenbaum2015strategyproof}, and sum of squared utilities \cite{ye2015strategy}. All of them are special cases of $L_p$ social utilities.  We present a complete picture for any value of $p$. In Section \ref{sec:norm} we study 
$L_p$ norms ($p\ge 1$), quasi norms ($0<p<1$), and the maximum utility ($p=+\infty$). In Section \ref{sec:min} we study the minimum utility ($p=-\infty$). 
In Section \ref{sec:nash} we consider the geometric mean and Nash welfare ($p\rightarrow 0^+$).

\subsection{$L_p$ Norms and Quasi Norms}\label{sec:norm}

For  the total  utility, i.e., $p=1$, Cheng et al.  \cite{cheng2013strategy} prove that Mechanism \ref{mec:mv}  (Majority Vote)  is 3-approximation, and Feigenbaum and  Sethuraman \cite{feigenbaum2015strategyproof} show the tightness of this bound for deterministic SP mechanisms.  For $p=2$, Ye et al. \cite{ye2015strategy} explore the sum of squared utilities and their results imply a $\sqrt 5$-approximation for the $L_2$ social utility. 
We extend these results to general $L_p$ norms ($p\ge 1$) and quasi norms ($0<p<1$).

\begin{theorem}\label{thm:1}
Mechanism~1 is a $(2^p+1)^{1/p}$-approximation for the $L_p$ social utility for every finite $p>0$, and is a 2-approximation for the $L_{\infty}$ social utility (max-utility).
\end{theorem}

\begin{proof}
By symmetry it suffices to analyze the case $n_1\le n_2$, in which Mechanism~1 outputs $y=0$. If $n_1=0$, then $y=0$ is optimal for every $p>0$, so we assume $n_1\ge 1$ in the sequel.

For profile $\mathbf x$, we write
\[
\mathrm{ALG}^p(\mathbf x)\coloneqq \sum_{i\in N} |x_i-0|^p.
\]
Let $\mathrm{OPT}^p(\mathbf x)\coloneqq \max_{y\in[0,1]} \left( \sum_{i\in N} |x_i-y|^p \right)$ denote the optimal $p$-power of the $L_p$ social utility. The approximation ratio is the $p$-th root of $\mathrm{OPT}^p(x)/\mathrm{ALG}^p(x)$.

We prove the bound by partitioning agents into disjoint pairs that cross the midpoint and (possibly) one unpaired agent on the side $x_i\ge\tfrac12$:
- For each pair $(i,j)$ with $x_i\le \tfrac12 < x_j$, we compare their contribution under $y=0$ against their optimal contribution under any $y$.
- For any agent $k$ with $x_k\ge \tfrac12$ that is not paired (this occurs only if $n_2>n_1$), we compare their singleton contribution under $y=0$ against the optimal singleton contribution.

The proof proceeds via the following two claims.

\medskip
\noindent\textbf{Claim 1 (Two-agent crossing pairs).}
Fix $x_i\le \tfrac12 < x_j$. Let
\[
A^p \coloneqq (x_i-0)^p + (x_j-0)^p
\quad\text{and}\quad
O^p \coloneqq \max_{y\in[0,1]} \bigl(|x_i-y|^p + |x_j-y|^p\bigr).
\]
Then $O^p \le 1 + \bigl(\tfrac12\bigr)^p \le (2^p+1) A^p$.

\emph{Proof of Claim 1.}
First, since $\tfrac12 < x_j$, we have $A^p\ge \bigl(\tfrac12\bigr)^p$.
We now upper bound $O^p$ by considering the location of an optimal $y^\ast$.

(i) If $y^\ast\le x_i$, then
\[
O^p \;=\; (x_i-y^\ast)^p + (x_j-y^\ast)^p \;\le\; x_i^p + x_j^p \;\le\; \Bigl(\tfrac12\Bigr)^p + 1.
\]

(ii) If $y^\ast\ge x_j$, by symmetry of (i), we again get $O^p \le 1 + \bigl(\tfrac12\bigr)^p$.

(iii) If $x_i<y^\ast<x_j$, then $O^p = (y^\ast-x_i)^p + (x_j-y^\ast)^p$ with $(y^\ast-x_i)+(x_j-y^\ast)=x_j-x_i\le 1$. For $0<p\le 1$, the function $t\mapsto t^p$ is concave, and for fixed sum the maximum of $a^p+b^p$ is attained at balance $a=b=\tfrac{x_j-x_i}{2}$, so
\[
O^p \;\le\; 2\Bigl(\tfrac{x_j-x_i}{2}\Bigr)^p \;\le\; 2\Bigl(\tfrac12\Bigr)^p \;<\; 1 + \Bigl(\tfrac12\Bigr)^p.
\]
For $p>1$, the function is convex and by the power-mean inequality, $(a+b)^p \ge a^p+b^p$ for $a,b\ge 0$, hence
\[
O^p \;=\; (y^\ast-x_i)^p + (x_j-y^\ast)^p \;\le\; (x_j-x_i)^p \;\le\; 1 \;<\; 1 + \Bigl(\tfrac12\Bigr)^p.
\]
Combining the three subcases yields $O^p \le 1 + \bigl(\tfrac12\bigr)^p$. Since $A^p\ge (\tfrac12)^p$, we obtain
\[
\frac{O^p}{A^p} \;\le\; \frac{1 + (\tfrac12)^p}{(\tfrac12)^p} \;=\; 2^p + 1,
\]
which proves the claim. 

\medskip
\noindent\textbf{Claim 2 (Unpaired agent on the right).}
If $x_k\ge \tfrac12$, then under $y=0$ the agent's contribution is $(x_k-0)^p$, and it is  optimal. 

The proof follows from $x_k\ge \tfrac12$ and $\arg\max_{y\in[0,1]} |x_k-y|^p=0$. 

\medskip
We now complete the proof by aggregating disjoint contributions. Form a disjoint collection $\mathcal{P}$ of $n_1$ crossing pairs $(i,j)$ (each with $x_i\le \tfrac12 < x_j$), and let $\mathcal{U}$ be the set of remaining unpaired agents located in $(\tfrac12,1]$ (there are $n_2-n_1$ such agents if any). By Claims~1 and~2, for each pair $(i,j)\in\mathcal{P}$,
\[
\max_{y} \bigl(|x_i-y|^p + |x_j-y|^p\bigr) \;\le\; (2^p+1)\bigl(x_i^p + x_j^p\bigr),
\]
and for each $k\in\mathcal{U}$,
\[
\max_{y} |x_k-y|^p \;\le\; (2^p+1)x_k^p.
\]
Summing these inequalities over the disjoint partition $(\mathcal{P},\mathcal{U})$ gives
\[
\mathrm{OPT}^p(\mathbf x)
\;\le\; (2^p+1)\sum_{i\in N} x_i^p
\;=\; (2^p+1)\,\mathrm{ALG}^p(\mathbf x).
\]
Taking the $p$-th root yields the approximation ratio $(2^p+1)^{1/p}$ for every finite $p>0$.

Finally, consider the limit $p\to\infty$ (max-utility). Since at least one agent lies in $(\tfrac12,1]$, the maximum utility under $y=0$ is at least $\tfrac12$, while the optimal maximum utility cannot exceed $1$. Hence the ratio is at most $2$, which coincides with $\lim_{p\to\infty} (2^p+1)^{1/p}=2$.
\end{proof}

We now show that this guarantee is tight among all deterministic strategyproof mechanisms.

\begin{theorem}\label{thm:de55}
No deterministic SP mechanism can achieve an approximation ratio better than $(2^p+1)^{1/p}$ for the $L_p$ social utility for any finite $p>0$, or an approximation ratio better than 2 for the $L_{\infty}$ social utility. 
\end{theorem}

\begin{proof}
By Lemma \ref{lemma:2-candidate}, any deterministic SP mechanism on the line either returns a fixed point or is a 2-candidate valid threshold mechanism. The first case is clearly unbounded. We show that the latter has a ratio at least $(2^p+1)^{1/p}$ for any finite positive $p$. 

Let the two candidates be $a\le b$. If $\frac12\le a\le b$, consider the profile $(a,b)$. The optimal solution is $0$, and the $p$-power of the optimal social utility is $\text{OPT}^p=a^p+b^p$. The mechanism selects either $a$ or $b$, and the $p$-power of the social utility is $\text{ALG}^p=0^p+(b-a)^p=(b-a)^p$. The ratio of them is $$\frac{\text{OPT}^p}{\text{ALG}^p}=\frac{a^p+b^p}{(b-a)^p}.$$
Let $d = b - a > 0$ and $t = \frac{a}{b-a} = \frac{a}{d}$.
Then
\[
\frac{a^p + b^p}{(b-a)^p}
= \frac{(td)^p + (td + d)^p}{d^p}
= t^p + (t+1)^p.
\]
From the constraints \(1 \ge b \ge a \ge \tfrac{1}{2}\), we have $d\le \frac12$ and $t=\frac ad\ge 1$.
Thus, the ratio is at least
$t^p + (t+1)^p \ge 1 + 2^p$.
 Symmetrically, if $a\le b\le \frac12$, the same profile $(a,b)$ forces a ratio at least $(2^p+1)^{1/p}$.
 
 Therefore, it suffices to consider the case with $a<\tfrac12<b$.
In our analysis we only consider the instances with $n=2$ agents in which no one is located at $\frac{a+b}{2}$ (with no ties). Then, a valid threshold mechanism must have a cutoff $s\in\{0,1,2\}$  such that the mechanism returns $a$ whenever the number of agents preferring $a$ is at least $s$, and returns $b$ otherwise. 
We treat the three possibilities separately.

\medskip
\noindent\emph{Case $s=0$.}
The mechanism always outputs $y=a$ for any profile under consideration. On the profile $(a,a)$, the mechanism's social utility is zero, while the optimal utility is strictly positive; the approximation ratio is unbounded.

\medskip
\noindent\emph{Case $s=1$.}
Consider the profile $\bigl(a,\tfrac{a+b}{2}+\varepsilon\bigr)$ with $\varepsilon>0$ arbitrarily small. The second agent prefers $a$, so the mechanism outputs $y=a$. Denote
\[
\mathrm{ALG}^p \;=\; |a-a|^p + \left|\tfrac{a+b}{2}+\varepsilon - a\right|^p
\;=\; \left(\tfrac{b-a}{2}+\varepsilon\right)^p.
\]
For the optimal solution,  placing the facility at $y=1$ gives
\[
\mathrm{OPT}^p \;\ge\; (1-a)^p + \left(1 - \tfrac{a+b}{2} - \varepsilon\right)^p.
\]
Letting $\varepsilon\to 0^+$ yields
\[
\frac{\mathrm{OPT}^p}{\mathrm{ALG}^p}
\ge \frac{(1-a)^p + \left(1-\tfrac{a+b}{2}\right)^p}{\left(\tfrac{b-a}{2}\right)^p}
= 2^p\frac{(1-a)^p}{(b-a)^p}+\frac{(2-a-b)}{(b-a)^p} \ge 2^p+1.
\]
Thus the ratio is at least $(2^p+1)^{1/p}$.

\medskip
\noindent\emph{Case $s=2$.}
This is symmetric to the case $s=1$. Taking the profile $\bigl(\tfrac{a+b}{2}-\varepsilon,\,b\bigr)$ with $\varepsilon\to 0^+$, only the first agent prefers $a$, and thus the mechanism outputs $y=b$. The same analysis as above yields the same lower bound $(2^p+1)^{1/p}$.

The proof for the $L_{\infty}$ utility is similar. When $\frac12\le a\le b$ or $a\le b\le \frac12$, the optimal maximum utility is $b$ or $1-a$, which is at least twice the maximum utility induced by the mechanism $b-a$. When $a<\frac12<b$, consider the cutoff $s\in\{0,1,2\}$. If $s=0$, the mechanism is unbounded for the profile $(a,a)$. If $s=1$, for the profile $(a,\frac{a+b}{2}+\epsilon)$, we have $\text{OPT}\ge 1-a$ and $\text{ALG}\to \frac{b-a}{2}\le \frac{\text{OPT}}{2}$. If $s=2$,  for the profile $(\frac{a+b}{2}-\epsilon,b)$, we have $\text{OPT}\ge b$ and $\text{ALG}\to \frac{b-a}{2}\le \frac{\text{OPT}}{2}$. 
\end{proof}

 The lower bound in Theorem~\ref{thm:de55} shows that deterministic mechanisms are fundamentally limited in their approximability. This motivates the study of randomization: can random choices, while preserving strategyproofness in expectation, yield strictly better worst-case guarantees? We next turn to randomized mechanisms and establish that they can outperform.

\begin{theorem}\label{thm:33}
 Mechanism~\ref{mec:mv-ran2} is a $2^{1/p}$-approximation for the $L_p$ social utility for every $0<p<1$, and a $(\frac{2^p+1}{2^{p-1}+1})^{1/p}$-approximation for every $1\le p<\infty$. For the $L_{\infty}$ social utility (max-utility), it is a $4/3$-approximation.
\end{theorem}

\begin{proof}
    
We first consider the $L_{\infty}$ social utility (max-utility). For any instance $\mathbf x=(x_1,\ldots,x_n)$, obviously an optimal solution  with the largest maximum utility must be $y^*\in \{0, 1\}$. When $n_1=0$ (resp. $n_2=0$), the mechanism returns the optimal solution $y=0$ (resp. $y=1$) with probability 1. When $n_1, n_2>0$, 
the probability of $y=0$ and $y=1$ is both $\frac{1}{2}$. Assume w.l.o.g. that $y^*=0$ is the optimal solution. The approximation ratio follows from
    \begin{align*}
        \frac{\text{OPT}}{\text{ALG}} = \frac{x_n-0}{\frac{1}{2}x_n+\frac{1}{2}(1-x_1)}\le \frac{1}{\frac{1}{2}+\frac{1}{2}(1-x_1)}\le \frac{1}{\frac{1}{2}+\frac{1}{4}} = \frac{4}{3}.
    \end{align*}

   For the case with finite positive $p$, if $n_1=0$ or $n_2=0$, again the mechanism returns the optimal solution deterministically. Assume $n_1,n_2>0$ and w.l.o.g. that  $0<n_1\le n_2$. We partition the agents into $n_1$ groups where each group consists of one agent at  $x_i\le \frac{1}{2}$ and $\frac{n_2}{n_1}$ agent(s) at the same point $x_j>\frac{1}{2}$, and we only need to prove the approximation ratio for each group. 
   Whenever $\frac{n_2}{n_1}$ is not integral  we can use $\lfloor\frac{n_2}{n_1}\rfloor$ or $\lceil\frac{n_2}{n_1}\rceil$ instead without affecting the analysis. 

  We explain why we can assume all the $\frac{n_2}{n_1}$ agents at the same location.
    Suppose there is a group in which one agent is at $x_i\le \frac12$ and the $\frac{n_2}{n_1}$ agents are located at $m$ different points $z_1,z_2,\ldots,z_m>\frac12$. For each $j\in [m]$, let $\lambda_j$ be the number of agents at $z_j$, and $\lambda_1+\lambda_2+\ldots+\lambda_m=\frac{n_2}{n_1}$. 
    If $y$ and $y^*$ are the solution of a mechanism and the optimal solution, respectively, denote by $a_j=|z_j-y|^p$ and $b_j=|z_j-y^*|^p$. 
    We write both $\text{ALG}^p$ and $\text{OPT}^p$ (defined within this group) as linear combinations of $a_j$ and $b_j$:  
    \begin{align*}
    \frac{\text{ALG}^p}{\text{OPT}^p}&=     \frac{\lambda_1a_1+\ldots+\lambda_ma_m+|x_i-y|^p}{\lambda_1b_1+\ldots+\lambda_mb_m+|x_i-y^*|^p}\\
    &\ge \min\left(\frac{a_1n_2/n_1+|x_i-y|^p}{b_1n_2/n_1+|x_i-y^*|^p},\ldots,\frac{a_mn_2/n_1+|x_i-y|^p}{b_mn_2/n_1+|x_i-y^*|^p}\right).
    \end{align*}
    Thus, there exists a location $z_k\in\{z_1,\ldots,z_m\}$ such that the ratio when all the $\frac{n_2}{n_1}$ agents are at $z_k$ is no better than the ratio for the group under consideration.

 From the perspective of worst-case analysis, thus we can safely consider a group with one agent at  $x_i\le \frac{1}{2}$ and $\frac{n_2}{n_1}$ agent(s) at the same location $x_j>\frac{1}{2}$. The $p$-power of the $L_p$ social utility induced by the mechanism for this group is  
    \begin{align*}
        \text{ALG}^p &= \frac{\left(n_2^2+2^pn_1n_2\right)\left(x_i^p+\frac{n_2}{n_1}x_j^p\right)}{n_1^2+n_2^2+2^{p+1}n_1n_2}+ \frac{\left(n_1^2+2^pn_1n_2\right)\left((1-x_i)^p+\frac{n_2}{n_1}(1-x_j)^p\right)}{n_1^2+n_2^2+2^{p+1}n_1n_2}.
    \end{align*}
    Then we discuss the optimal social utility OPT. There are three cases for the optimal solution $y^*$. 
    
\noindent(i) $y^*\le x_i$, which means $y^*=0$.  The $p$-power of the optimal social utility is
    \begin{align*}
        \text{OPT}^p = x_i^p+\frac{n_2}{n_1}x_j^p
    \end{align*}
    and
    \begin{align*}
        \frac{\text{ALG}^p}{\text{OPT}^p} &= \frac{n_2^2+2^pn_1n_2}{n_1^2+n_2^2+2^{p+1}n_1n_2}+\frac{n_1^2+2^pn_1n_2}{n_1^2+n_2^2+2^{p+1}n_1n_2}\cdot \frac{n_1(1-x_i)^p+n_2(1-x_j)^p}{n_1x_i^p+n_2x_j^p}\\
        &\ge \frac{n_2^2+2^pn_1n_2}{n_1^2+n_2^2+2^{p+1}n_1n_2}+\frac{n_1^2+2^pn_1n_2}{n_1^2+n_2^2+2^{p+1}n_1n_2}\cdot \frac{n_1(1-\frac{1}{2})^p+n_2(1-1)^p}{n_1 (\frac{1}{2})^p+n_2\cdot 1^p}\\
        &=1 - \frac{n_1^2+2^pn_1n_2}{n_1^2+n_2^2+2^{p+1}n_1n_2}\cdot\frac{2^pn_2}{n_1+2^pn_2}\\
        &= 1 - \frac{2^pn_1n_2}{n_1^2+n_2^2+2^{p+1}n_1n_2}\ge 1-\frac{2^p}{2^{p+1}+2} = \frac{2^{p-1}+1}{2^p+1},
    \end{align*}
which is larger than $\frac12$. Hence, this gives a $2^{1/p}$-approximation when $0<p<1$ and a $\left(\frac{2^p+1}{2^{p-1}+1}\right)^{1/p}$-approximation when $p\ge 1$. 

\noindent(ii) $y^*\ge x_j$, which means $y^*=1$.  The $p$-power of the optimal social utility is
    \begin{align*}
        \text{OPT}^p = (1-x_i)^p+\frac{n_2}{n_1}(1-x_j)^p.
    \end{align*}
As the proof in (i) is independent of the condition $n_1\le n_2$, the symmetric argument of (i)   proves this case. 

\noindent(iii) $x_i< y^*< x_j$. The $p$-power of the optimal social utility is
    \begin{align*}
        \text{OPT}^p&=(y^*-x_i)^p+\frac{n_2}{n_1}(x_j-y^*)^p.
    \end{align*}

For $p\ge 1$, by the super-additivity of function $x^p$, we have
    \begin{align*}
        \text{OPT}^p \le \frac{n_2}{n_1}\left((y^*-x_i)^p+(x_j-y^*)^p\right)\le \frac{n_2}{n_1}(x_j-x_i)^p\le x_i^p+\frac{n_2}{n_1}x_j^p.
    \end{align*}
   Therefore, the same argument as in (i) establishes the proof. 

    For $0<p<1$, we consider a function $f(z)$ with $z\in (0, x_j-x_i)$:
    \begin{align*}
        f(z) = n_1z^p+n_2(x_j-x_i-z)^p.
    \end{align*}
    The derivative is
    \begin{align*}
        f'(z)=pn_1z^{p-1}-pn_2(x_j-x_i-z)^{p-1}.
    \end{align*}
    Solving the equation $f'(z)=0$, we have
    \begin{align*}
        z= \frac{n_2^{1/(p-1)}(x_j-x_i)}{n_1^{1/(p-1)} + n_2^{1/(p-1)}},
    \end{align*}
    in which $f(z)$  reaches the maximum value.
    Hence, $\text{OPT}^p$ reaches the maximum value when
    \begin{align*}
        y^* = x_i + (x_j-x_i)\cdot \frac{n_2^{\frac{1}{p-1}}}{n_1^{\frac{1}{p-1}} + n_2^{\frac{1}{p-1}}} = x_j - (x_j-x_i)\cdot \frac{n_1^{\frac{1}{p-1}}}{n_1^{\frac{1}{p-1}} + n_2^{\frac{1}{p-1}}}
    \end{align*}
    and
    \begin{align*}
        \text{OPT}^p &= \left((x_j-x_i)\frac{n_2^{\frac{1}{p-1}}}{n_1^{\frac{1}{p-1}} + n_2^{\frac{1}{p-1}}}\right)^p + \frac{n_2}{n_1}\left((x_j-x_i)\frac{n_1^{\frac{1}{p-1}}}{n_1^{\frac{1}{p-1}} + n_2^{\frac{1}{p-1}}}\right)^p\\
        &= \frac{n_2\cdot(x_j-x_i)^p}{\left(n_1^{\frac{1}{p-1}}+n_2^{\frac{1}{p-1}}\right)^{p-1}}\cdot \left(\frac{n_2^{\frac{1}{p-1}}}{n_1^{\frac{1}{p-1}} + n_2^{\frac{1}{p-1}}}+\frac{n_1^{\frac{1}{p-1}}}{n_1^{\frac{1}{p-1}} + n_2^{\frac{1}{p-1}}}\right) \\
        &= \frac{n_2(x_j-x_i)^p}{\left(n_1^{\frac{1}{p-1}}+n_2^{\frac{1}{p-1}}\right)^{p-1}}\le \frac{n_2(x_j-x_i)^p}{\left(n_1^{-1}+n_2^{-1}\right)^{-1}} = \frac{(n_1+n_2)}{n_1}(x_j-x_i)^p.
    \end{align*}
 Here we explain the inequality. Let $q=-\frac{1}{p-1}\in [1, +\infty)$, and the denominator of the LHS of the inequality is 
    $\left(\left(\frac{1}{n_1}\right)^q+\left(\frac{1}{n_2}\right)^q\right)^{-1/q}$, which is increasing with respect to $q$. Hence, the denominator attains its minimum value when $q=1$ and $p=0$, that is, $\left(n_1^{-1}+n_2^{-1}\right)^{-1}$. 
    
    Then the ratio is
    \begin{align*}
        \frac{\text{ALG}^p}{\text{OPT}^p}&\ge \frac{\frac{\left(n_2^2+2^pn_1n_2\right)\left(x_i^p+\frac{n_2}{n_1}x_j^p\right)}{n_1^2+n_2^2+2^{p+1}n_1n_2} + \frac{\left(n_1^2+2^pn_1n_2\right)\left((1-x_i)^p+\frac{n_2}{n_1}(1-x_j)^p\right)}{n_1^2+n_2^2+2^{p+1}n_1n_2}}
        {\frac{(n_1+n_2)}{n_1}(x_j-x_i)^p}.
    \end{align*}
   We claim that when $x_i=0, x_j=1$, the expression reaches the minimum value. Due to symmetry, we only need to prove that for any fixed $x_i$,  the value of the expression, denoted by $h(x_j)$,  decreases with respect to $x_j\in [x_i, 1]$. We have $h(x_j)$ is proportional to
    \begin{align*}
        &\frac{\left(n_2^2+2^pn_1n_2\right)\left(x_i^p+\frac{n_2}{n_1}x_j^p\right)+\left(n_1^2+2^pn_1n_2\right)\left((1-x_i)^p+\frac{n_2}{n_1}(1-x_j)^p\right)}{(x_j-x_i)^p}.
    \end{align*}
     Let $u(x_j)$ denote the numerator and $v(x_j)=(x_j-x_i)^p$ denote the denominator. We have
    \begin{align*}
        h'(x_j)&\propto u'(x_j)v(x_j)-u(x_j)v'(x_j)\\
        &\propto \frac{1}{p}\cdot u'(x_j)(x_j-x_i)-u(x_j)\\
        &\le \frac{n_2}{n_1}\left(n_2^2+2^pn_1n_2\right)(x_j-x_i)x_j^{p-1} - \frac{n_2}{n_1}\left(n_2^2+2^pn_1n_2\right)x_j^p\\
        &\le 0.
    \end{align*}
    Therefore, the claim holds and we have
    \begin{align*}
        \frac{\text{ALG}^p}{\text{OPT}^p}
        &\ge \frac{\frac{\left(n_2^2+2^pn_1n_2\right)\left(0^p+\frac{n_2}{n_1}1^p\right)}{n_1^2+n_2^2+2^{p+1}n_1n_2} + \frac{\left(n_1^2+2^pn_1n_2\right)\left((1-0)^p+\frac{n_2}{n_1}(1-1)^p\right)}{n_1^2+n_2^2+2^{p+1}n_1n_2}}
        {\frac{(n_1+n_2)}{n_1}(1-0)^p}\\
        &= \frac{n_1}{n_1+n_2}\cdot\left(\frac{\left(n_2^2+2^pn_1n_2\right)\frac{n_2}{n_1}+\left(n_1^2+2^pn_1n_2\right)}{n_1^2+n_2^2+2^{p+1}n_1n_2}\right)\\
        &= \frac{n_2^3+2^pn_1n_2^2+2^pn_1^2n_2+n_1^3}{n_2^3+(2^{p+1}+1)n_1n_2^2+(2^{p+1}+1)n_1^2n_2+n_1^3}.
    \end{align*}
    Let $\alpha := \frac{n_2}{n_1}\ge 1$. It follows that
    \begin{align*}
        \frac{\text{ALG}^p}{\text{OPT}^p}&\ge \frac{\alpha^3+2^p\alpha^2+2^p\alpha+1}{\alpha^3+(2^{p+1}+1)\alpha^2+(2^{p+1}+1)\alpha+1}\\
        &= \frac{1}{2} + \frac{\frac{1}{2}
        \alpha^3-\frac{1}{2}\alpha^2-\frac{1}{2}\alpha+\frac{1}{2}}{\alpha^3+(2^{p+1}+1)\alpha^2+(2^{p+1}+1)\alpha+1}\\
        &= \frac{1}{2}\left(1+\frac{(\alpha-1)(\alpha^2-1)}{\alpha^3+(2^{p+1}+1)\alpha^2+(2^{p+1}+1)\alpha+1}\right)\\
        &\ge \frac{1}{2}.
    \end{align*}
    Therefore, the approximation ratio for $0<p<1$ is $2^{1/p}$. 
\end{proof}

It is worth noting that the limit $\lim_{p\to \infty}\left(\frac{2^p+1}{2^{p-1}+1}\right)^{1/p}=1$ is not equal to the approximation ratio $\frac43$ for the $L_{\infty}$ utility. This is because for randomized mechanisms, the limit of expectations and the expectation of limits may not coincide. We give an example to show the approximation ratio $\frac43$ is attainable: consider a profile $(0.5,1)$, for which the optimal facility location is $0$ and the optimal maximum utility is $1$. Mechanism~\ref{mec:mv-ran2} returns $0$ and $1$ with probability $\frac12$ each, and the maximum utility is $\frac12\cdot 1+\frac12\cdot\frac12=\frac34$.

We complement the results by a lower bound of randomized mechanisms.

\begin{theorem}\label{thm:ljh}
    No  randomized SP mechanism can achieve an approximation ratio better than $r$ for the $L_p$ social utility for any finite $p>0$, where
    \begin{align*}
        r = \left(\frac{4(3^p+1)}{3(3^p+1)+2}\right)^\frac{1}{p}.
    \end{align*}
    No randomized SP mechanism can achieve  an approximation ratio better than $\frac{6}{5}$ for the $L_{\infty}$ social utility.
\end{theorem}

\begin{proof}
    Let $f$ be a randomized SP mechanism.  
Consider profile $\mathbf{x}_1 = (\frac{1}{3}, \frac{2}{3})$. Obviously for any output $y$, we have $\left|y-\frac{1}{3}\right|+\left|y-\frac{2}{3}\right|\le 1$. W.l.o.g. suppose $\mathbb{E}_{y_1\sim f(\mathbf x_1)}\left[\left|y_1-\frac{2}{3}\right|\right]\le \frac{1}{2}$. Then we consider a new profile $\mathbf{x}_2 = (\frac{1}{3}, 1)$. Due to SP property, we must  have $\mathbb{E}_{y_2\sim f(\mathbf x_2)}\left[\left|y_2-\frac{2}{3}\right|\right]\le \frac{1}{2}$, otherwise at profile $\mathbf x_1$, agent 2 can misreport from $\frac{2}{3}$ to 1. The optimal solution is $y^*=0$ and $\text{OPT}^p = 1+ \left(\frac{1}{3}\right)^p$.

    Then we consider the mechanism's social utility under the distribution $y_2\sim f(\mathbf x_2)$. We want to find the largest possible social utility under the constraint $\mathbb{E}\left[\left|y_2-\frac{2}{3}\right|\right]\le \frac{1}{2}$. 

 \textbf{$0<p<1$. } We show that the largest possible social utility is attained when $y_2=0$ with probability $\frac34$ and $y_2=\frac23$ with probability $\frac14$. 
First, for any point $y\in [\frac{1}{3}, 1]$ with some probability, replacing it by $y'=\frac{2}{3}$ with the same probability can only decrease the expected distance to $\frac23$, and will not decrease the social utility as well: Observe that the expression $x^p + (1-x)^p$ with $0\le x\le 1$ reaches the maximum value when $x=\frac{1}{2}$ for $0<p<1$. Therefore for any such $y$, letting $y=\frac{2x+1}{3}$, the social utility satisfies
    \begin{align*}
        &~\left(y-\frac{1}{3}\right)^p+\left(1-y\right)^p\le \left(\frac{2}{3}-\frac{1}{3}\right)^p + \left(1-\frac{2}{3}\right)^p.
    \end{align*}

    Second, 
for any point $y\in [0, \frac{1}{3})$ with some probability $P_0$, replace it by $y'=0$ with probability $(1-\frac{3}{2}y)\cdot P_0$ and $y'=\frac{2}{3}$ with probability $\frac{3}{2}y\cdot P_0$. It is easy to see that this maintains the expected distance to $\frac23$. To show that it will not decrease the social utility, we only need to prove
    \begin{align*}
        \left(\frac{1}{3}-y\right)^p + (1-y)^p\le \left(1-\frac{3}{2}y\right)\cdot \left(\left(\frac{1}{3}\right)^p+1^p\right) + \frac{3}{2}y\cdot 2\cdot \left(\frac{1}{3}\right)^p.
    \end{align*}
    Let $g(y)$ denote the difference between the LHS and the RHS. The derivative is
    \begin{align*}
        g'(y) &= \frac{3}{2}\cdot \left(\left(\frac{1}{3}\right)^p+1^p\right)-3\cdot \left(\frac{1}{3}\right)^p-\left(\frac{1}{3}-y\right)^{p-1} - (1-y)^{p-1}\\
        &= \frac{3}{2}\left(1-\frac{1}{3^p}\right)-\left(\frac{1}{3}-y\right)^{p-1} - (1-y)^{p-1}\\
        &\le \frac{3}{2}\left(1-\frac{1}{3}\right)-1\le 0,
    \end{align*}
    indicating that $g(y)$ decreases with $y$, and the maximum value of $g(y)$ is $g(0)=0$.
 
    Therefore,  the largest possible social utility is attained  when $y_2=0$ with probability $q$ and $y_2=\frac{2}{3}$ with probability $1-q$. Maximizing this social utility under the constraint $\mathbb{E}\left[\left|y_2-\frac{2}{3}\right|\right]\le \frac{1}{2}$ gives $q=\frac34$. We have
    \begin{align}
        \frac{\text{OPT}^p}{\text{ALG}^p} &\ge \frac{1 + \left(\frac{1}{3}\right)^p}{\frac{3}{4}\cdot \left(1 + \left(\frac{1}{3}\right)^p\right) + \frac{1}{4}\cdot 2\cdot \left(\frac{1}{3}\right)^p} = \frac{4(3^p+1)}{3(3^p+1)+2}.\label{eq:678}
    \end{align}

    \textbf{$1<p<\infty$.} Again, we show that the largest possible social utility is
attained when $y_2=0$ with probability $\frac{3}{4}$ and $y_2=\frac{2}{3}$ with probability $\frac{1}{4}$. First,  for any point $y\ge \frac{2}{3}$ with some probability, replacing it by $y'=\frac{4}{3}-y\in [\frac{1}{3}, \frac{2}{3}]$ with the same probability will not affect the social utility as well as the expected distance to $\frac23$, because the two points $y$ and $y'$ are symmetric.

Second, for any point $y=\frac{2}{3}-c\in [\frac{1}{3}, \frac{2}{3}]$ ($c\in [0, \frac{1}{3}]$) with probability $P_0$ and any $c'\in [c, \frac{1}{3}]$, let $\alpha=1-\frac{c}{c'}$. Then replace $y$ by $y'=\frac{2}{3}-c'$ with probability $P_0\cdot(1-\alpha)$ and $y'=\frac{2}{3}$ with probability $P_0\cdot\alpha$. It is clear that the expected distance to $\frac23$ will not increase.  To show that the social utility will not decrease, it suffices to prove 
    \begin{align*}
        &\left(\frac{1}{3}+c\right)^p + \left(\frac{1}{3}-c\right)^p \le2\alpha\left(\frac{1}{3}\right)^p
        + (1-\alpha)\left(\left(\frac{1}{3}+\frac{c}{1-\alpha}\right)^p+\left(\frac{1}{3}-\frac{c}{1-\alpha}\right)^p\right),
    \end{align*}
    which means that we only need to prove
    \begin{align*}
        \left(\frac{1}{3}\pm c\right)^p\le \alpha\left(\frac{1}{3}\right)^p + (1-\alpha)\left(\frac{1}{3}\pm \frac{c}{1-\alpha}\right)^p.
    \end{align*}
    It follows immediately from the convexity of the function $x^p$ when $p>1$.

Third, for any point $y=\frac{1}{3}-c\in [0, \frac{1}{3}]$ ($c\in [0, \frac{1}{3}]$) with probability $P_0$ and any $c''\in [c, \frac{1}{3}]$, let $\alpha=1-\frac{c}{c''}$. Then replace $y$ by $y'=\frac{1}{3}-c''$ with probability $P_0\cdot(1-\alpha)$ and $y'=\frac{1}{3}$ with probability $P_0\cdot\alpha$. It is clear that the expected distance to $\frac23$ will not increase.
To show that the social utility will not decrease, it suffices to prove
    \begin{align*}
        &c^p + \left(\frac{2}{3}+c\right)^p \le\alpha\cdot \left(0^p + \left(\frac{2}{3}\right)^p\right)
        + (1-\alpha)\left(\left(\frac{c}{1-\alpha}\right)^p+\left(\frac{2}{3}+\frac{c}{1-\alpha}\right)^p\right),
    \end{align*}
    which still holds because of the convexity of $x^p$ when $p>1$. 

    Combining the above three claims, we can move all $y\in [\frac{2}{3}, 1]$ to $\frac{4}{3}-y\in [\frac{1}{3}, \frac{2}{3}]$, and then move all $y\in [\frac{1}{3}, \frac{2}{3}]$ to $\frac{2}{3}-c'$ and $\frac{2}{3}$ with specific probabilities for $c'=\frac13$, and move all $y\in [0, \frac{1}{3}]$  to $\frac{1}{3}-c''$ and $\frac{2}{3}$ with specific probabilities for  $c''=\frac13$. This gives a randomized 3-candidate mechanisms with $y\in \{0, \frac{1}{3}, \frac{2}{3}\}$.

    The final step is to move the point $y=\frac{1}{3}$ with probability $P_0$ to $y'=0$ and $y'=\frac{2}{3}$ with probability $\frac{P_0}{2}$ each. While the expected distance to $\frac23$ maintains the same, the social utility also will not decrease because 
    \begin{align*}
        0^p + \left(\frac{2}{3}\right)^p \le \frac12\cdot\left( \left(\frac{1}{3}\right)^p+1^p+\left(\frac13\right)^p+\left(\frac{2}{3}\right)^p\right).
    \end{align*}
  Therefore,  the largest possible social utility is attained  when $y_2=0$ with probability $q$ and $y_2=\frac{2}{3}$ with probability $1-q$. Maximizing this social utility under the constraint $\mathbb{E}\left[\left|y_2-\frac{2}{3}\right|\right]\le \frac{1}{2}$ gives $q=\frac34$. Then the lower bound follows f
  rom Eq. \eqref{eq:678}. 

\textbf{$p=+\infty$.} A similar proof as the case when $1<p<\infty$ gives that the largest possible $L_{\infty}$ utility is attained  when $y_2=0$ with probability $\frac34$ and $y_2=\frac{2}{3}$ with probability $\frac14$. Thus the ratio is
    \begin{align*}
        \frac{\text{OPT}}{\text{ALG}} = \frac{1}{\frac{3}{4}\cdot 1+\frac{1}{4}\cdot \frac{1}{3}} = \frac{6}{5}.
    \end{align*}
\end{proof}

We remark that when $p=2$, our lower bound is $\frac{\sqrt{5}}{2}\approx 1.118$, which improves the previous result of $1.042^{1/2}\approx 1.021$ in \cite{ye2015strategy}.

\subsection{Minimum Utility}\label{sec:min}

When $p=-\infty$, the $L_p$ social utility is the minium utility of agents, and the minimum-utility-maximization is the egalitarian objective.   
We first compute the optimal minimum utility, and then derive upper and lower bounds. For convenience, let $x_1\le x_2\le \ldots \le x_n$ in profile $\mathbf x$. The following proposition is clear.

\begin{proposition}
    The optimal minimum utility is 
    $$\text{OPT}_{-\infty}(\mathbf x)=\max\left(x_1,1-x_n, \max_{1\le i\le n-1}\frac{x_{i+1}-x_i}{2}\right).$$ 
    If $\text{OPT}_{-\infty}(\mathbf x)=x_1$, then $y=0$ is an optimal solution; if $\text{OPT}_{-\infty}(\mathbf x)=1-x_n$, then $y=1$ is an optimal solution; if $\text{OPT}_{-\infty}(\mathbf x)=\frac{x_{k+1} - x_k}{2}$ for some $k$, then $\frac{x_k+x_{k+1}}{2}$ is an optimal solution.
\end{proposition}

For deterministic mechanisms, the approximation ratio is unbounded, as shown by Feigenbaum and Sethuraman \cite{feigenbaum2015strategyproof}. The good news is that randomized mechanisms can achieve a bounded approximation guarantee. However, no 2-candidate randomized mechanism suffices: whenever the two candidates are $\{a,b\}$, the ratio remains unbounded for the profile $(a,b)$ as the minimum utility is 0. By the same reasoning, no $k$-candidate randomized mechanism can provide a bounded approximation. We show that the uniform distribution has a sub-linear approximation. 

\begin{theorem}
    Mechanism \ref{mec:uniform} is a $\mathcal{O}(\sqrt{n})$-approximation for the minimum utility.
\end{theorem}

\begin{proof}
    For any profile $\mathbf x$, Mechanism \ref{mec:uniform} always returns the uniform distribution over interval $[0,1]$. The minimum utility is
    \begin{align*}
        \text{ALG} &= \int_0^1 \min_{i\in N} |y-x_i|\cdot dy\\
        &= x_1\cdot \frac{x_1}{2}+\sum_{i=1}^{n-1}(x_{i+1}-x_i)\cdot \frac{x_{i+1}-x_i}{4}+(1-x_n)\cdot \frac{1-x_n}{2}.
    \end{align*}
    We consider three cases based on the optimum.

    \textbf{Case 1}: $\text{OPT}_{-\infty}(\mathbf x)=x_1$. According to Cauchy-Schwarz Inequality, we have
    \begin{align*}
        \sum_{i=1}^{n-1}(x_{i+1}-x_i)^2 + (1-x_n)^2&\ge \frac{1}{n}\left(1-x_n+\sum_{i=1}^{n-1}(x_{i+1}-x_i)\right)^2\\
        &= \frac{(1-x_1)^2}{n}.
    \end{align*}
    Therefore, the ratio is 
    \begin{align*}
        \frac{\text{OPT}_{-\infty}}{\text{ALG}} &\le \frac{x_1}{\frac{x_1^2}{2} + \frac{(1-x_1)^2}{4n}} = \frac{4n}{2nx_1+\frac{(1-x_1)^2}{x_1}} \le \frac{4n}{2nx_1+\frac{1}{x_1}-2}\\
        &\le \frac{4n}{2\sqrt{2n}-2} = \mathcal{O}(\sqrt{n}).
    \end{align*}

    \textbf{Case 2}: $\text{OPT}_{-\infty}(\mathbf x)=1-x_n$. By symmetry, the proof is the same as in Case 1 and is omitted. 

    \textbf{Case 3}: $\text{OPT}_{-\infty}(\mathbf x)=\frac{x_{k+1}-x_k}{2}$ for some $k$. According to Cauchy-Schwarz Inequality, we have 
    \begin{align*}
        &x_1^2 + \sum_{i=1}^{k-1}(x_{i+1}-x_i)^2 + \sum_{i=k+1}^{n-1}(x_{i+1}-x_i)^2 + (1-x_n)^2\\
        &\ge \frac{1}{n}\left(x_1+\sum_{i=1}^{k-1}(x_{i+1}-x_i) + \sum_{i=k+1}^{n-1}(x_{i+1}-x_i) + (1-x_n)\right)^2\\
        &= \frac{(x_k+1-x_{k+1})^2}{n}.
    \end{align*}
    Denote by $a = x_{k+1}-x_k\in [0, 1]$. We have
    \begin{align*}
        \text{ALG} \ge  \frac{(1-a)^2}{4n} + \frac{x_1^2+(1-x_n)^2}{4}+\frac{a^2}{4}\ge \frac{(1-a)^2}{4n}+\frac{a^2}{4}.
    \end{align*}
    Therefore, the ratio is    \begin{align*}
        \frac{\text{OPT}_{-\infty}}{\text{ALG}} &\le \frac{\frac{a}{2}}{\frac{(1-a)^2}{4n} + \frac{a^2}{4}} = \frac{2n}{na+\frac{(1-a)^2}{a}}\le \frac{2n}{na+\frac{1}{a}-2} \le \frac{2n}{2\sqrt{n}-2} = \mathcal{O}(\sqrt{n}).
    \end{align*}
\end{proof}

The following example shows that the analysis of Mechanism \ref{mec:uniform} is asymptotically tight. Let $x_1=\frac{1}{\sqrt{n}}$, and let $x_2, \ldots, x_n$ evenly partition the interval $[x_1, 1]$. The optimal minimum utility is $\frac{1}{\sqrt{n}}$, while $\text{ALG}=\frac{1}{2n} + \frac{n+1}{4n^2}\cdot \left(1-\frac{1}{\sqrt{n}}\right)^2$. Hence the approximation ratio is $\Omega(\sqrt{n})$.

As for the inapproximability of randomized mechanisms, \cite{feigenbaum2015strategyproof} proves a lower bound of 1.5 in the limit $n\rightarrow \infty$. This does not directly apply  when $n$ is finite, particularly for small $n$. Below we provide a lower bound for the case of two agents.

\begin{theorem}\label{thm:n2}
    No randomized SP mechanism has an approximation ratio better than $1.026$ for the minimum utility when $n=2$. 
\end{theorem}
\begin{proof}
    We provide our core idea for this theorem. Consider a two-agent profile $\mathbf{x} = (x_1, x_2) = (0, 1)$. This time $\text{OPT}_1=\frac{1}{2}$ with solution $y=\frac{1}{2}$. Suppose $\epsilon>0$ is a tiny constant (e.g. 0.001) and we want to prove the lower bound is at least $1+\epsilon$. Due to symmetry, we suppose $\mathbb{E}[y]\ge \frac{1}{2}$.

    In this way, we move $x_2$ to some $x_2'<x_2$ but $x_2-x_2'$ is also tiny. In order to maintain the approximation ratio of the mechanism, we can find a constant $\delta>0$ which is also tiny so that any mechanism must output $y\in [\frac{x_2'}{2}-\delta, \frac{x_2'}{2}+\delta]$ with probability at least $P_0$ which is close to 1 (e.g. 0.99). Obviously if $x_2'$ is smaller, than $\delta$ is also smaller. 

    We start from calculating the relationship between $\delta$ and $\epsilon$. According to the approximation ratio of $(x_1, x_2')$, we have
    \begin{align*}
        1+\epsilon = \frac{\text{OPT}_2}{\text{ALG}_2}\ge \frac{\frac{x_2'}{2}}{P_0\cdot \frac{x_2'}{2}+(1-P_0)\cdot (\frac{x_2'}{2}-\delta)}.
    \end{align*}
    Solving this equation, we have
    \begin{align*}
        \epsilon \ge \frac{x_2'}{x_2'-2\delta(1-P_0)} - 1 = \frac{2\delta(1-P_0)}{x_2'-2\delta(1-P_0)}.
    \end{align*}
    If we want to have the largest lower bound, i.e. the largest $\epsilon$, then we need to let $P_0, x_2'$ be smaller and $\delta$ be larger.
    
    Notice that when agent 2 moves to $x_2'$, we must guarantee $1-x_2'<\frac{x_2'}{2}-\delta$ because the minimum utility of solution $y'=1$ cannot be the same as (or better than) solution $y'=\frac{x_2'}{2}\pm \delta$. Therefore we have $x_2'>\frac{2(1+\delta)}{3}$ and for the largest lower bound, we let $x_2'\rightarrow \frac{2(1+\delta)}{3}$. Therefore we have
    \begin{align*}
        \epsilon\ge \frac{2\delta(1-P_0)}{\frac{2(1+\delta)}{3}-2\delta(1-P_0)}.
    \end{align*}

    According to SP property, in case agent 2 misreports from $x_2=1$ to $x_2'$, we must have $\mathbb{E}[y']\ge \mathbb{E}[y]$. However, notice that
    \begin{align*}
        \mathbb{E}[y']\le P_0\cdot \left(\frac{x_2'}{2}+\delta\right) + (1-P_0)\cdot 1.
    \end{align*}
    Therefore if the right part of the inequality above is less than $\frac{1}{2}$, then we get into a contradiction. This means when it is equal to $\frac{1}{2}$, we get a lower bound. Then we need to find the largest such $\epsilon$. We aim to prove that if there is a mechanism better than $(1+\epsilon)$-approximation, then it will be contradicted with SP property. Therefore
    \begin{align*}
        P_0\cdot \frac{1+4\delta}{3} + 1-P_0=1-\frac{2-4\delta}{3}P_0<\frac{1}{2}.
    \end{align*}
    Solving this and we have
    \begin{align*}
        P_0 > \frac{3}{4(1-2\delta)}
    \end{align*}
    and the right part is the smallest $P_0$. So $\epsilon$ can be re-written as
    \begin{align*}
        \epsilon\ge \frac{2\delta\cdot \frac{1-8\delta}{4-8\delta}}{\frac{2+2\delta}{3} - 2\delta\cdot \frac{1-8\delta}{4-8\delta}} = \frac{\frac{2+2\delta}{3}}{\frac{2+2\delta}{3}-2\delta\cdot \frac{1-8\delta}{4-8\delta}}-1.
    \end{align*}

    By drawing the graph of the function, when $\delta\approx 0.065153$, $\epsilon$ reaches the maximum value $\epsilon\approx0.025909\approx0.026$.
\end{proof}

\subsection{Geometric Mean}\label{sec:nash}

 When $p\rightarrow 0^+$, the unnormalized $\text{su}_p$ becomes unbounded, and we consider the normalized $L_p$ social utility instead, that is, the geometric mean of utilities. Formally, the objective is to maximize $\sqrt[n]{\prod_{i\in N} u(x_i, y)}$, which is equivalent to maximizing the Nash welfare $\prod_{i\in N} u(x_i, y)$. 
First we show that, the approximation ratio is unbounded for all deterministic mechanisms.

\begin{theorem}
    No deterministic SP mechanism has bounded approximation ratio for the  geometric mean of utilities.
\end{theorem}
\begin{proof}
    Recall that any deterministic SP mechanism has at most two candidates by Lemma \ref{lemma:2-candidate}. For the profile where the agents are located at these candidates,  then the social utility of the mechanism is 0, while the optimal social utility is positive. 
\end{proof}

Hence we consider randomized mechanisms and show that the uniform distribution (Mechanism \ref{mec:uniform}) is a constant approximation. 

\begin{theorem}
    Mechanism \ref{mec:uniform} is a $(\sqrt{2}+1)$-approximation for the geometric mean of utilities.
\end{theorem}

\begin{proof}
    We only need to prove that, for each agent, the individual utility induced  by the mechanism is a $(\sqrt{2}+1)$-approximation of their best possible utility. For each agent $i$ with $x_i\ge \frac{1}{2}$, the best possible  utility is equal to $x_i$, and her utility under the mechanism is 
    \begin{align*}
    u(x_i,\mathcal U([0,1])) &=\! \int_{0}^{1}|y-x_i| dy = \!\int_{0}^{x_i} (x_i-y) dy +\! \int_{x_i}^{1} (y-x_i) dy\\
        &= x_i \frac{x_i}{2} + (1-x_i)\frac{1-x_i}{2} = \frac{2x_i^2-2x_i+1}{2}.
    \end{align*}
The ratio is 
    \begin{align*}
        \frac{x_i}{u(x_i,\mathcal U([0,1]))} = \frac{x_i}{\frac{2x_i^2-2x_i+1}{2}} = \frac{2}{2x_i+\frac{1}{x_i}-2} \le \frac{2}{2\sqrt{2}-2} = \sqrt{2}+1.
    \end{align*}
    For an agent $i$ with $x_i< \frac{1}{2}$, the symmetric proof gives the ratio. 
\end{proof}

\begin{theorem}
    No randomized SP mechanism has an approximation ratio better than $\sqrt{\frac{6}{5}}$ for the geometric mean of utilities.
\end{theorem}

\begin{proof}
The proof idea mirrors that of Theorem \ref{thm:ljh}. Let $f$ be a randomized SP mechanism.   
Consider the profile $\mathbf{x}_1 = (\frac{1}{3}, \frac{2}{3})$. For any output $y$, we have $\left|y-\frac{1}{3}\right|+\left|y-\frac{2}{3}\right|\le 1$. W.l.o.g. assume $\mathbb{E}_{y_1\sim f(\mathbf x_1)}\left[\left|y_1-\frac{2}{3}\right|\right]\le \frac{1}{2}$. Now we consider a new profile $\mathbf{x}_2 = (\frac{1}{3}, 1)$. By strategyproofness, we must have $\mathbb{E}_{y_2\sim f(\mathbf x_2)}\left[\left|y_2-\frac{2}{3}\right|\right]\le \frac{1}{2}$, otherwise at profile $\mathbf x_1$, agent 2 (whose true location is $\frac{2}{3}$) could beneficially misreport  1. The optimal solution is $y^*=0$, yielding $\text{OPT}^2 = \frac13\cdot 1$.

    Next, evaluate the mechanism’s social utility under $y_2\sim f(\mathbf x_2)$. We seek the largest possible social utility subject to the constraint $\mathbb{E}\left[\left|y_2-\frac{2}{3}\right|\right]\le \frac{1}{2}$. Following the analysis for the case $0<p<1$ in the proof of Theorem \ref{thm:ljh}, the extreme distribution is attained by placing probability $\frac{3}{4}$ on $y_2=0$ and probability $\frac{1}{4}$ on $y_2=\frac{2}{3}$. 
Therefore, the approximation ratio satisfies
    \begin{align*}
        \frac{\text{OPT}^2}{\text{ALG}^2} \ge \frac{\frac{1}{3}\cdot 1}{\frac{3}{4}\cdot \frac{1}{3}\cdot 1 + \frac{1}{4}\cdot \frac{1}{3}\cdot \frac{1}{3}} = \frac{6}{5}.
    \end{align*}
\end{proof}

\section{$L_p$ Social Costs}\label{sec:sc}

When the facility is obnoxious, agents prefer it to be as far away as possible. Modeling each agent's cost as the remaining segment length—1 minus their distance to the facility—captures this preference: the closer the facility, the higher the cost.
In this section, we study the $L_p$ social cost objectives.
For $p<1$, the $L_p$  social cost is a non-norm and non-convex aggregator that overly rewards spreading small costs while tolerating very large costs for a few agents, violating fairness and robustness. Practically, this ``risk-seeking" behavior and loss of geometric/algorithmic properties make it ill-suited for cost-minimization objectives.
Therefore,  this section focuses on $L_p$ social costs with  $1\le p<\infty$, as well as on the limiting case $p=\infty$ (the max-cost).  In Section \ref{sec:deter}, we derive tight bounds for deterministic mechanisms.  In Section \ref{sec:rand} we show that randomized mechanisms can improve the approximation ratio for the utilitarian objective ($p=1$) and establish lower bounds for randomized mechanisms.

\subsection{Deterministic Mechanisms}\label{sec:deter}

The {\rm Majority Vote} mechanism (Mechanism \ref{mec:mv}) performs well for the $L_p$ social cost objectives when $p\ge 1$. 

\begin{theorem}
    Mechanism \ref{mec:mv}  is a $\left(2^p+1\right)^{\frac{1}{p}}$-approximation for the $L_p$ social cost for every $1\le p<\infty$. For the $L_{\infty}$ social cost (max-cost), it is a 2-approximation. 
\end{theorem}

\begin{proof}
   When $p=+\infty$, it is the egalitarian objective that minimizes the maximum cost. If $n_1=0$ (resp. $n_2=0$), the solution $y=0$ (resp. $y=1$) returned by the mechanism is clearly optimal for all agents. 
   If both $n_1, n_2> 0$,  the maximum cost induced by the mechanism is no more than 1, while  the optimal maximum cost is at least $\frac{1}{2}$. The 2-approximation immediately follows.

 When $1\le p<\infty$,  due to symmetry, we only consider the case when $n_1\le n_2$ and $y=0$. Note that for each agent $i$ with $x_i\ge \frac{1}{2}$, the solution $y=0$ is the best possible for them. Thus, we can assume $n_1=n_2$ without hurting the performance guarantee. We partition the agents into pairs $(i,j)$ with $x_i\in [0, \frac{1}{2}]$ and $x_j\in (\frac{1}{2}, 1]$. It suffices to prove that $y=0$ is a $(2^p+1)$-approximation for the $p$-power of the $L_p$ social cost of these two agents. 

The $p$-power of the mechanism's social cost for $(i,j)$ is 
    \begin{align*}
        A^p = (1-x_i)^p + (1-x_j)^p \le 1 + \frac{1}{2^p}.
    \end{align*}
    Then, denote by $O^p$ the $p$-power of the social cost of $(i,j)$ under the optimal solution $y^*$. We only need to prove the following inequality:
    \begin{align*}
        O^p  \ge \frac{1}{2^p} = \frac{1 + \frac{1}{2^p}}{2^p+1} \ge \frac{\text{ALG}^p}{2^p+1}.
    \end{align*}
    If $y^*\le x_i$, then
    \begin{align*}
        O^p = (1-(x_i-y^*))^p+(1-(x_j-y^*))^p\ge \left(\frac{1}{2}\right)^p + 0 = \left(\frac{1}{2}\right)^p.
    \end{align*}
    If $y^*\ge x_j$, it is symmetric and the proof is omitted.
    If $y^*\in (x_i, x_j)$, then
    \begin{align*}
       O^p &= (1-(y^*-x_i))^p+(1-(x_j-y^*))^p\ge 2\cdot\left(1-\frac{x_j-x_i}{2}\right)^p\\
        &\ge 2\cdot \frac{1}{2^p} > \frac{1}{2^p},
    \end{align*}
   where the first inequality is because the function $g(x)=x^p+(c_0-x)^p$ with constants $c_0>0$ and $p\ge 1$ reaches the minimum value when $x=\frac{c_0}{2}$. 
    
   Therefore,  we have  $A^p\le (2^p+1)\cdot O^p$ for each pair $(i,j)$. Summing up the $p$-power of the social cost over all pairs, we obtain
   $\text{ALG}^p\le (2^p+1)\cdot\text{OPT}^p$, establishing the approximation ratio. 
\end{proof}

The following lower bound result indicates that {\rm Majority Vote} is indeed the best possible for deterministic mechanisms.

\begin{theorem}\label{thm:sc-deter-lower}
    No deterministic SP mechanism has an approximation ratio better than $\left(2^p+1\right)^{\frac{1}{p}}$ for the $L_p$ social cost for any $1\le p<\infty$.  No deterministic SP mechanism has an approximation ratio better than 2 for the $L_{\infty}$ social cost.  
\end{theorem}

\begin{proof}
    We consider a two-agent profile $\mathbf{x}_1 = (0, 0)$. The optimal solution is $y=1$ with social cost 0. Any mechanism must also output $y_1=1$ to guarantee a finite approximation ratio. In the same way, when $\mathbf{x}_2 = (1,1)$, any mechanism must output $y_2=0$. Since any SP mechanism has at most  two candidates by Lemma \ref{lemma:2-candidate}, it can only output $y\in \{0, 1\}$ for all profiles.

    Then we consider profile $\mathbf{x}_3=(0, 1)$ and assume w.l.o.g. $y_3=0$. Let agent 2 moves to $\frac{1}{2}+\epsilon$ with sufficiently small $\epsilon > 0$ and the profile becomes $\mathbf{x}_4=(0, \frac{1}{2}+\epsilon)$. The output must be $y_4=0$, otherwise agent 2 at $\mathbf x_4$ can misreport from $\frac{1}{2}+\epsilon$ to 1 to decrease the cost. Then the approximation ratio for $1\le p<\infty$ approaches
    \begin{align*}
        \frac{\text{ALG}^p}{\text{OPT}^p} = \frac{1^p + \left(\frac{1}{2}-\epsilon\right)^p}{\left(\frac{1}{2}+\epsilon\right)^p} \rightarrow \frac{1+\left(\frac{1}{2}\right)^p}{\left(\frac{1}{2}\right)^p} = 2^p+1,
    \end{align*}
    and the ratio for $p= \infty$ approaches
    \begin{align*}
        \frac{\text{ALG}}{\text{OPT}} = \frac{1}{\frac{1}{2}+\epsilon}\rightarrow 2.
    \end{align*}
\end{proof}

\subsection{Randomized Mechanisms}\label{sec:rand}

For randomized mechanisms, we show that  Mechanism \ref{mec:mv-ran} (which  returns $0$ with probability $\frac{n_2^2}{n_1^2+n_2^2}$ and $1$ with probability $\frac{n_1^2}{n_1^2+n_2^2}$) is 2-approximation for both the utilitarian ($p=1$) and egalitarian ($p=+\infty$) objectives. In particular, when $p=1$ it improves the bound 3  for deterministic mechanisms.  

\begin{theorem}\label{thm:sc-ran-upper}
    Mechanism \ref{mec:mv-ran} is a $2$-approximation for the $L_p$ social cost with both $p=1$ and $p=+\infty$.
\end{theorem}

\begin{proof}
When $p=+\infty$,  if $n_1=0$ (resp. $n_2=0$), the solution $y=0$ (resp. $y=1$) deterministically returned by the mechanism is clearly optimal for all agents. 
   If both $n_1, n_2> 0$,  the maximum cost induced by the mechanism is no more than 1, while  the optimal maximum cost is at least $\frac{1}{2}$. This indicates a 2-approximation.
   
  When $p=1$, due to symmetry, we assume $n_1\le n_2$. If $n_1=0$, the mechanism deterministically returns the optimal  solution $y=0$. Assume $n_1>0$. As in the proof of Theorem \ref{thm:33}, we partition the agents into $n_1$ groups where each group consists of one agent at  $x_i\le \frac{1}{2}$ and $\frac{n_2}{n_1}$ agent(s) at the same point $x_j>\frac{1}{2}$. 
   Whenever $\frac{n_2}{n_1}$ is not integral  we can use $\lfloor\frac{n_2}{n_1}\rfloor$ or $\lceil\frac{n_2}{n_1}\rceil$ instead without affecting the analysis.  We only need to prove the approximation ratio for each group. 

 For such a group, the social cost induced by the mechanism is
    \begin{align*}
        \text{ALG} = \frac{n_1^2}{n_1^2+n_2^2}\left(x_i+\frac{n_2}{n_1}x_j\right) + \frac{n_2^2}{n_1^2+n_2^2}\left((1-x_i)+\frac{n_2}{n_1}(1-x_j)\right).
    \end{align*}
    Then we discuss the optimal social cost for  this group. There are three candidates that may be the optimal solution $y^*$. 
    
    \textbf{Case 1}: $y^*\ge x_j$, which means $y^*=1$. The optimal social cost is
    \begin{align*}
        \text{OPT} = x_i+\frac{n_2}{n_1}x_j
    \end{align*}
    and
    \begin{align*}
        \frac{\text{ALG}}{\text{OPT}} &= \frac{n_1^2}{n_1^2+n_2^2}+\frac{n_2^2}{n_1^2+n_2^2}\cdot \frac{n_1(1-x_i)+n_2(1-x_j)}{n_1x_i+n_2x_j}\\
        &\le \frac{n_1^2}{n_1^2+n_2^2}+\frac{n_2^2}{n_1^2+n_2^2}\cdot \frac{n_1(1-0)+n_2(1-\frac{1}{2})}{n_1 \cdot 0+n_2\cdot(\frac{1}{2})}\\
        &=1 + \frac{n_2^2}{n_1^2+n_2^2}\cdot\frac{2n_1}{n_2}=1+\frac{2n_1n_2}{n_1^2+n_2^2}\\
        &\le 2.
    \end{align*}

    \textbf{Case 2}: $y^*<x_i$, which means $y^*=0$. 
    This is symmetric to case 1 and we omit the proof.

    \textbf{Case 3}: $x_i\le y^*\le x_j$.
   The optimal social cost is
    \begin{align*}
        \text{OPT}&=(1-(y^*-x_i))+\frac{n_2}{n_1}(1-(x_j-y^*)).
    \end{align*} 
   It reaches the minimum value when $y^*=x_i$, indicating that $ \text{OPT} = 1+\frac{n_2}{n_1}(1-x_j+x_i)$. Then the ratio is
    \begin{align*}
        \frac{\text{ALG}}{\text{OPT}} &= \frac{\frac{n_1^2}{n_1^2+n_2^2}\left(x_i+\frac{n_2}{n_1}x_j\right) + \frac{n_2^2}{n_1^2+n_2^2}\left((1-x_i)+\frac{n_2}{n_1}(1-x_j)\right)}{1+\frac{n_2}{n_1}(1-x_j+x_i)}\\
        &\le \frac{\frac{n_1^2}{n_1^2+n_2^2}\left(0+\frac{n_2}{n_1}\cdot 1\right) + \frac{n_2^2}{n_1^2+n_2^2}\left((1-0)+\frac{n_2}{n_1}(1-1)\right)}{1+\frac{n_2}{n_1}(1-1)}\\
        &= \frac{n_1n_2}{n_1^2+n_2^2} + \frac{n_2^2}{n_1^2+n_2^2} \le 2.
    \end{align*}

    Therefore, in all of the three cases the approximation ratio is at most 2. 
\end{proof}

After conducting numerical experiments, we conjecture that  Mechanism \ref{mec:mv-ran} is indeed a $2$-approximation for any $p\ge 1$.

Finally, we present lower bounds for randomized strategyproof mechanisms.

\begin{theorem}\label{thm:fin}
    No randomized SP mechanism has an approximation ratio better than $(\frac{5}{4})^{1/p}$ for the $L_p$ social cost for any $1\le p<\infty$.
    No randomized SP mechanism has an approximation ratio better than $1.008$ for  the $L_{\infty}$ social cost.
\end{theorem}

\begin{proof}
Let $f$ be a randomized SP mechanism. When  $1\le p<\infty$,  consider profile $\mathbf{x}_1 = (\frac{1}{3}, \frac{2}{3})$. For any output $y_1$, we have $\left|y_1-\frac{1}{3}\right|+\left|y_1-\frac{2}{3}\right|\le 1$. W.l.o.g. suppose $\mathbb{E}\left[\left|y_1-\frac{2}{3}\right|\right]\le \frac{1}{2}$. Then we consider a new profile $\mathbf{x}_2 = (\frac{1}{3}, 1)$. By strategyproofness, we must have $\mathbb{E}\left[\left|y_2-\frac{2}{3}\right|\right]\le \frac{1}{2}$, otherwise agent 2 at $\mathbf x_1$ can misreport from $\frac{2}{3}$ to 1. The optimal solution is $y^*=0$ and $\text{OPT}^p = 0 + (\frac{2}{3})^p = (\frac{2}{3})^p$.

    Then we consider the mechanism's social cost for profile $\mathbf{x}_2 = (\frac{1}{3}, 1)$.  We seek the smallest possible social cost subject to the constraint $\mathbb{E}\left[\left|y_2-\frac{2}{3}\right|\right]\le \frac{1}{2}$. 

   First, for any point $y\in [\frac{1}{3}, 1]$ with some probability, we replace $y$ by $y'=\frac{2}{3}$ with the same probability. We show that the replacement will not increase the social cost. Note that function $f(x) = x^p + (1-x)^p$ with $0\le x\le 1$ reaches the minimum value when $x=\frac{1}{2}$ for $p>1$. Hence for any such $y$, letting $y=\frac{4x}{3}$ we have
    \begin{align*}
        &~(1-(y-\frac{1}{3}))^p+(1-(1-y))^p\ge (1-(\frac{2}{3}-\frac{1}{3}))^p + (1-(1-\frac{2}{3}))^p,
    \end{align*}
    where the LHS is the $p$-power of the social cost under $y$, and the RHS is that under $y'$.

Second, for any point $y\in [0, \frac{1}{3}]$ with some probability $P_0$, we replace $y$ by returning $y'=0$ with probability $(1-\frac{3}{2}y)\cdot P_0$ and $y'=\frac{2}{3}$ with probability $\frac{3}{2}y\cdot P_0$. To show that the replacement will not increase the social cost, it suffices to prove
    \begin{align*}
        (y+\frac{2}{3})^p + y^p&\ge (1-\frac{3}{2}y)\cdot (\frac{2}{3})^p + \frac{3}{2}y\cdot 2\cdot (\frac{2}{3})^p\\
        &= (1+\frac{3}{2}y)\cdot (\frac{2}{3})^p.
    \end{align*}
    We only need to prove     
    \begin{align*}
        (y+\frac{2}{3})^p\ge (1+\frac{3}{2}y)\cdot (\frac{2}{3})^p,
    \end{align*}
    which holds because
        $(y+\frac{2}{3})^{p-1}\ge(\frac{2}{3} )^{p-1}$ when $p\ge 1$.
 
    Therefore, the smallest social cost under the distance constraint is attained when returning $y=0$ with probability $\frac{3}{4}$ and $y=\frac{2}{3}$ with probability $\frac{1}{4}$. The approximation ratio satisfies
    \begin{align*}
        \frac{\text{ALG}^p}{\text{OPT}^p} &\ge \frac{\frac{3}{4}\cdot (\frac{2}{3})^p + \frac{1}{4}\cdot 2\cdot (\frac{2}{3})^p}{(\frac{2}{3})^p} = \frac{5}{4}.
    \end{align*}

    When $p=+\infty$,  the objective is to minimize the maximum cost. We use the same arguments as in the proof of Theorem \ref{thm:n2} (for the minimum utility) with $\delta\approx 0.026$, $P_0 \approx 0.79$, $x_2'\approx 0.684$. Then the ratio satisfies
    \begin{align*}
        r = \frac{1-\text{ALG}_2}{1-\text{OPT}_2}\approx  \frac{1-0.336554}{1-0.342}\approx 1.008.
    \end{align*}
\end{proof}

Furthermore, we establish lower bounds for \emph{two-candidate randomized} mechanisms~\cite{li2024strategyproof}: there exist two candidates \(\{a,b\}\) such that, for any profile \(\mathbf{x}\), the mechanism outputs \(y=a\) with probability \(P_0(\mathbf{x})\) and \(y=b\) with probability \(1-P_0(\mathbf{x})\).
Next we provide a lower bound of 2-candidate randomized mechanism.

\begin{theorem}
    No SP $(a,b)$-candidate randomized mechanism can achieve better than $(2^{p-1}+1)^{\frac{1}{p}}$-approximation for the $L_p$ social cost for any $p\ge 1$.
\end{theorem}
\begin{proof}
    If all the agents are at 0, then the optimal solution is $y=1$ and the social cost is 0. Therefore $1\in \{a, b\}$. In the same way, $0\in \{a, b\}$. Therefore we know all bounded mechanisms are $(0, 1)$-candidate randomized mechanisms. Consider a profile $\mathbf{x}_1=(0, 1)$ and suppose the probability of $y=0$ is no more than $\frac{1}{2}$ due to symmetry. Then we move agent 1 to $\frac{1}{2}-\epsilon$ where $\epsilon\rightarrow 0$ is positive. This time the probability of $y=0$ cannot increase, otherwise agent 1 can misreport from $\frac{1}{2}-\epsilon$ to 0 in order to decrease the cost. This time the optimal solution is $y=0$ with $\text{OPT}^p = (\frac{1}{2}+\epsilon)^p$.

    Denoting the probability of $y=0$ by $P_0$ for profile $\mathbf{x}_2 = (\frac{1}{2}-\epsilon, 1)$ with $P_0\le \frac{1}{2}$, we have
    \begin{align*}
        \frac{\text{ALG}^p}{\text{OPT}^p}&=\frac{P_0\cdot \left(\frac{1}{2}+\epsilon\right)^p + (1-P_0)\cdot \left(1 + \left(\frac{1}{2}-\epsilon\right)^p\right)}{\left(\frac{1}{2}+\epsilon\right)^p}\\
        &\ge \frac{\frac{1}{2}\cdot \left(\frac{1}{2}+\epsilon\right)^p + (1-\frac{1}{2})\cdot \left(1 + \left(\frac{1}{2}-\epsilon\right)^p\right)}{\left(\frac{1}{2}+\epsilon\right)^p}\\
        &\rightarrow \frac{(\frac{1}{2})^{p+1} + \frac{1}{2} + (\frac{1}{2})^{p+1}}{(\frac{1}{2})^p} = 1 + 2^{p-1}.
    \end{align*}

    When $p=\infty$, i.e. maximum cost, we consider profile $\mathbf{x}_1=(0, 1)$. The optimal solution is $y=\frac{1}{2}$ with $\text{OPT}=\frac{1}{2}$, and the mechanism's maximum cost is 1, indicating an approximation of at least 2.
\end{proof}

\section{Conclusion}

We provide a unified analysis of the approximation guarantees of various (group) SP mechanisms for locating a single obnoxious facility on a bounded unit interval to approximately optimize the family of \(L_p\)-aggregated utility and cost objectives for $p \in (-\infty, \infty)$.  
For \(L_p\)-aggregated utility and cost objectives, we provide upper and lower bounds on the approximation ratios of any deterministic and randomized (group) SP mechanisms. 
Our bounds for deterministic (group) SP mechanisms are tight, while our bounds for randomized (group) SP mechanisms have gaps. 
However, randomized (group) SP mechanisms can often achieve better approximation ratios. 

Future work includes tightening gaps for randomized (group) SP mechanisms, extending to higher-dimensional or networked domains, and incorporating richer fairness and robustness notions. 

\paragraph{Acknowledgements.} Hau Chan is supported by the National Institute of General Medical Sciences of the National Institutes of Health [P20GM130461], the Rural Drug Addiction Research Center at the University of Nebraska-Lincoln, and the National Science Foundation under grants IIS:RI \#2302999 and IIS:RI \#2414554. The work is supported in part by the Guangdong Provincial/Zhuhai Key Laboratory of IRADS (2022B1212010006) and by Artificial Intelligence and Data Science Research Hub, BNBU, No. 2020KSYS007. Chenhao Wang is supported by BNBU under grant UICR0400004-24B. The content is solely the responsibility of the authors and does not necessarily represent the official views of the funding agencies.

\bibliographystyle{plain}
\bibliography{mybibfile}

\end{document}